\documentclass[a4paper,UKenglish]{lipics-v2016}
 
\usepackage{microtype}

\title{Clifford algebras meet tree decompositions\footnote{
This work is partially supported by Foundation for Polish Science grant HOMING PLUS/2012-6/2
and a project TOTAL that has received funding from the European Research Council (ERC)
under the European Union’s Horizon 2020 research and innovation programme (grant agreement No 677651).}}

\author[1]{Micha{\l} W{\l}odarczyk}
\affil[1]{University of Warsaw, Faculty of Mathematics, Informatics, and Mechanics, Warsaw, Poland\\
  \texttt{m.wlodarczyk@mimuw.edu.pl}}
\authorrunning{M. W{\l}odarczyk} 

\Copyright{Micha{\l} W{\l}odarczyk}

\subjclass{F.2.2 [Nonnumerical Algorithms and Problems]: Computations on discrete structures;
G.2.2 [Graph Theory]: Graph algorithms}
\keywords{fixed-parameter tractability, treewidth, Clifford algebra, algebra isomorphism}

\usepackage{anyfontsize} 

\newcommand{\ba}{\backslash}
\newcommand{\re}{\mathbb{R}}
\newcommand{\ze}{\mathbb{Z}}
\newcommand{\clifr}{$Cl_{n,0}(\mathbb{R})\,$}
\newcommand{\clifz}{$Cl_{n,0}(\mathbb{Z})\,$}
\newcommand{\llarrow}{\xrightarrow{\hspace*{3cm}}}
\newcommand\myatop[2]{\genfrac{}{}{0pt}{}{#1}{#2}}
\theoremstyle{plain}
\newtheorem{claim}[theorem]{Claim}
\newtheorem{observation}[theorem]{Observation}

\newcommand{\defproblem}[3]{
  \vspace{1mm}
\noindent\fbox{
  \begin{minipage}{0.9\textwidth}
  #1 \\
  {\bf{Input:}} #2  \\
  {\bf{Decide:}} #3
  \end{minipage}
  }
}

\EventEditors{John Q. Open and Joan R. Acces}
\EventNoEds{2}
\EventLongTitle{42nd Conference on Very Important Topics (CVIT 2016)}
\EventShortTitle{CVIT 2016}
\EventAcronym{CVIT}
\EventYear{2016}
\EventDate{December 24--27, 2016}
\EventLocation{Little Whinging, United Kingdom}
\EventLogo{}
\SeriesVolume{42}
\ArticleNo{23}

\begin{document}

\maketitle

\begin{abstract}
We introduce the Non-commutative Subset Convolution --
a convolution of functions useful when working with determinant-based algorithms.
In order to compute it efficiently, we take advantage of Clifford algebras,
a generalization of  quaternions used mainly in the  quantum field theory.

We apply this tool to speed up algorithms counting subgraphs parameterized by the treewidth of a graph.
We present an $O^*((2^\omega + 1)^{tw})$-time algorithm for counting Steiner trees
and an $O^*((2^\omega + 2)^{tw})$-time algorithm for counting Hamiltonian cycles,
both of which improve the previously known upper bounds.
The result for \textsc{Steiner Tree} also translates into a deterministic algorithm for \textsc{Feedback Vertex Set}.
All of these constitute the best known running times of deterministic algorithms for decision versions of these problems
and they match the best obtained running times
for pathwidth parameterization under assumption $\omega = 2$.
\end{abstract}

\section{Introduction}

The concept of \textit{treewidth} has been introduced by Robertson and Seymour in their work 
on graph minors \cite{treewidth}.
The treewidth of a graph is the smallest possible width of its \textit{tree decomposition}, i.e. a tree-like representation of the graph.
Its importance follows from the fact that many NP-hard graph problems become solvable
on trees with a simple dynamical programming.
A similar idea of \textit{pathwidth} captures the width of a graph in case we would like to have a \textit{path decomposition}.
Formal definitions can be found in Section \ref{sec:treewidth}.

A bound on the graph's treewidth allows to design efficient algorithms using \textit{fixed-parameter tractability}.
An algorithm is called fixed-parameter tractable (FPT) if it works in time complexity $f(k)n^{O(1)}$ where $k$ is a parameter
describing hardness of the instance and $f$ is a computable function.
We also use notation $O^*(f(k))$ that suppresses polynomial factors with respect to the input size.
Problems studied in this work are parameterized by the graph's pathwidth or treewidth.
To distinguish these cases we denote the parameter respectively $pw$ or $tw$.

It is natural to look for a function $f$ that is growing relatively slow.
For problems with a local structure, like \textsc{Vertex Cover} or \textsc{Dominating Set}, there are simple FPT
algorithms with single exponential running time.
They usually store $c^{tw}$ states for each node of the decomposition
and take advantage of the Fast Subset Convolution \cite{fsc} to perform the
\textit{join} operation in time $O^*(c^{tw})$.
As a result, time complexities for pathwidth and treewidth parameterizations remain the same.
The Fast Subset Convolution turned out helpful in many other problems, e.g. \textsc{Chromatic Number},
and enriched the basic toolbox used for exponential and parameterized algorithms.

Problems with connectivity conditions, like \textsc{Steiner Tree} or \textsc{Hamiltonian Cycle},
were conjectured to require time $2^{\Omega(tw\log tw)}$ until the breakthrough work of Cygan et al.~\cite{single}.
They introduced the randomized technique \textsc{Cut \& Count} working in single exponential time.
The obtained running times were respectively $O^*(3^{tw})$ and $O^*(4^{tw})$.
Afterwards, a faster randomized algorithm for \textsc{Hamiltonian Cycle} parameterized by the pathwidth was presented
with running time $O^*((2 + \sqrt{2})^{pw})$ \cite{matchings}.
This upper bound as well as $O^*(3^{pw})$ for \textsc{Steiner Tree}
are tight modulo subexponential factors under the assumption of
\textit{Strong Exponential Time Hypothesis} \cite{matchings, single}.

The question about the existence of single exponential deterministic methods was answered positively by Bodlaender et al. \cite{det}.
What is more, presented algorithms count the number of Steiner trees or Hamiltonian cycles in a graph.
However, in contrast to the Cut \& Count technique, a large gap emerged between the running times for pathwidth and treewidth parameterizations --
the running times were respectively $O^*(5^{pw})$, $O^*(10^{tw})$ for \textsc{Steiner Tree}
and $O^*(6^{pw})$, $O^*(15^{tw})$ for \textsc{Hamiltonian Cycle}.
This could be explained by a lack of efficient algorithms to perform the \textit{join} operation, necessary only for tree decompositions.
Some efforts have been made to reduce this gap and the deterministic running time for \textsc{Steiner Tree}
has been improved to $O^*((2^{\omega - 1} \cdot 3 + 1)^{tw})$~\cite{records}.

\subsection{Our contribution}

The main contribution of this work is creating a link between Clifford algebras,
objects not being used in algorithmics to the best of our knowledge, and fixed-parameter tractability.
As the natural dynamic programming approach on tree decompositions uses the Fast Subset Convolution (FSC)
to perform efficiently the \textit{join} operation, there
was no such a tool for algorithms based on the determinant approach.

Our first observation is that the FSC technique can be regarded as
an isomorphism theorem for some associative algebras.
To put it briefly, a Fourier-like transform is being performed in the FSC to bring computations to a simpler algebra.
Interestingly, this kind of transform is just a special case of the Artin-Wedderburn theorem \cite{artin}, 
which seemingly is not widely reported in computer science articles.
The theorem provides a classification of a large class of associative algebras,
not necessarily commutative (more in Appendix \ref{app:algebra}).
We use this theory to introduce the Non-commutative Subset Convolution (NSC) and speed up multiplication operations in an algebra
induced by the \textit{join} operation in determinant-based dynamic programming on tree decomposition.
An important building block is a fast Fourier-like transform
for a closely related algebra \cite{isomorphism}.
We hope our work will encourage researchers to investigate further algorithmic applications of the Artin-Wedderburn theorem.

\subsection{Our results}

We apply our algebraic technique to determinant approach introduced by Bodlaender et al.~\cite{det}.
For path decomposition, they gave an $O^*(5^{pw})$-time algorithm for counting Steiner trees and an $O^*(6^{pw})$-time algorithm for counting Hamiltonian cycles.
The running times for tree decomposition were respectively $O^*(10^{tw})$ and $O^*(15^{tw})$.
These gaps can be explained by the appearance of the \textit{join} operation in tree decompositions
which could not be handled efficiently so far.

By performing NSC in time complexity $O^*(2^\frac{\omega n}{2})$ we partially solve an open problem about the \textit{different convolution}
from~\cite{open}.
Our further results may be considered similar to those closing the gap between time complexities for pathwidth and treewidth
parameterizations for \textsc{Dominating Set} by switching between representations of states in dynamic programming~\cite{rooij}.
We improve the running times to $O^*((2^\omega + 1)^{tw})$ for counting Steiner trees and  $O^*((2^\omega + 2)^{tw})$ for counting Hamiltonian cycles,
where $\omega$ denotes the matrix multiplication exponent (currently it is established that $\omega < 2.373$~\cite{omega}).
These are not only the fastest known algorithms for counting these objects, but also the fastest known deterministic algorithms for
the decision versions of these problems.
The deterministic algorithm for \textsc{Steiner Tree} can be translated into a deterministic algorithm for \textsc{Feedback Vertex Set}~\cite{det}
so our technique provides an improvement also in this case.

Observe that the running times for pathwidth and treewidth parameterizations match under the assumption $\omega = 2$.
Though we do not hope for settling the actual value of $\omega$, this indicates there is no further
space for significant improvement unless pure combinatorial algorithms (i.e. not based on matrix multiplication) are invented
or the running time for pathwidth parameterization is improved.

\subsection{Organization of the paper}

Section \ref{sec:clif} provides a brief introduction to Clifford algebras.
The bigger picture of the employed algebraic theory can be found in Appendix \ref{app:algebra}.
In Section \ref{sec:nsc} we define the NSC
and design efficient algorithms for variants of the NSC employing the algebraic tools.
Sections \ref{sec:count-st} and \ref{sec:count-hc} present how to apply the NSC
in counting algorithms for \textsc{Steiner Tree} and \textsc{Hamiltonian Cycle}.
They contain main ideas improving the running times, however in order
to understand the algorithms completely one should start from Section 4 (\textit{Determinant approach}) in \cite{det}.
The algorithm for \textsc{Hamiltonian Cycle} is definitely more complicated
and its details, formulated as two isomorphism theorems, are placed in Appendix \ref{app:hc}.

\section{Preliminaries}\label{sec:prelim}

We will start with notation conventions.

\begin{enumerate}
\item $A \uplus B = C$ stands for $(A \cup B = C) \land (A \cap B = \emptyset)$.
\item $A \triangle B = (A \ba B) \cup (B \ba A)$.
\item $[\alpha]$ equals 1 if condition $\alpha$ holds and 0 otherwise.
\item For permutation $f$ of a linearly ordered set $U$
\begin{equation*}
\text{sgn}(f) = (-1)^{|\{(a,b) \in U \times U \,\land\, a < b \,\land\, f(a) > f(b)\}|}.
\end{equation*}
\item For $A,B$ being subsets of a linearly ordered set
\begin{equation}\label{eq:i}
I_{A,B} = (-1)^{|\{(a,b) \in A \times B \,\land\, a > b\}|}.
\end{equation}
\end{enumerate}
Let us note two simple properties of $I$.

\begin{claim}\label{lem:i1}
For disjoint $A,B$
\begin{equation*}
I_{A,B}I_{B,A} = (-1)^{|A||B|}.
\end{equation*}
\end{claim}

\begin{claim}\label{lem:i2}
For $A\cap B = \emptyset$ and $C\cap D = \emptyset$
\begin{equation*}
I_{A\cup B, C\cup D} = I_{A,C}I_{A,D}I_{B,C}I_{B,D}.
\end{equation*}
\end{claim}

\subsection{Fast Subset Convolution}

Let us consider a universe $U$ of size $n$ and functions $f,g:2^U \longrightarrow \ze$.

\begin{definition}
The M\"{o}bius transform of $f$ is function $\hat{f}$ defined as
\begin{equation*}
\hat{f}(X) = \sum_{A \subseteq X}f(A).
\end{equation*}
\end{definition}

\begin{definition}
Let $f * g$ denote a \textit{subset convolution} of $f,g$ defined as
\begin{equation*}
(f * g)(X) = \sum_{A \uplus B = X}f(A)g(B).
\end{equation*}
\end{definition}

\begin{theorem}[Bj\"{o}rklund et al. \cite{fsc}]\label{thm:fsc}
The M\"{o}bius transform, its inverse, and the subset convolution can be computed in time $O^*(2^n)$.
\end{theorem}

\subsection{Pathwidth and treewidth}\label{sec:treewidth}

\begin{definition}
A \textit{tree (path) decomposition} of a graph $G$ is a tree $\mathbb{T}$ (path $\mathbb{P}$)
in which each node $x$ is assigned a bag $B_x \subseteq V(G)$ such that
\begin{enumerate}
\item for every edge $uv \in E(G)$ there is a bag $B_x$ containing $u$ and $v$,
\item for every vertex $v$ the set $\{x\, |\, v \in B_x\}$ forms a non-empty subtree (subpath) in the decomposition.
\end{enumerate}
The width of the decomposition is defined as $\max_x |B_x| - 1$ and the treewidth (pathwidth) of $G$
is a minimum width over all possible tree (path) decompositions.
\end{definition}

If a graph admits a tree decomposition of width $t$ then it can be found in time $n\cdot2^{O(t^3)}$~\cite{bodlaender-linear}
and a decomposition of width at most $4t+1$ can be constructed in time $poly(n)\cdot2^{O(t)}$~\cite{kloks}.
We will assume that a decomposition of the
appropriate type and width is given as a part of the input.

\begin{definition}
A \textit{nice tree (path) decomposition} is a decomposition
with one special node $r$ called the root and in which each bag is one of the following types:
\begin{enumerate}
\item \textit{Leaf bag:} a leaf $x$ with $B_x = \emptyset$,
\item \textit{Introduce vertex $v$ bag:} a node $x$ having one child $y$ for which
$B_x = B_y \uplus \{v\}$,
\item \textit{Forget vertex $v$ bag:} a node $x$ having one child $y$ for which
$B_y = B_x \uplus \{v\}$,
\item \textit{Introduce edge $uv$ bag:} a node $x$ having one child $y$ for which
$u,v \in B_x = B_y$,
\item \textit{Join bag:} \textbf{(only in tree decomposition)} a node $x$ having two children $y,z$ with condition $B_x = B_y = B_z$.
\end{enumerate}
We require that every edge from $E(G)$ is introduced exactly once and $B_r$ is an empty bag.
For each $x$ we define $V_x$ and $E_x$ to be sets of respectively vertices and edges introduced
in the subtree of the decomposition rooted at $x$.
\end{definition}

Given a tree (path) decomposition we can find a nice decomposition in time $n\cdot tw^{O(1)}$~\cite{single, kloks}
and we will work only on these.
When analyzing running time of algorithms working on tree decompositions
we will estimate the bag sizes from the above assuming $|B_x| = tw$.

\subsection{Problems definitions}

\defproblem{\textsc{Steiner Tree}}
{graph $G$, set of terminals $K \subseteq V(G)$, integer $k$}
{whether there is a subtree of $G$ with at most $k$ edges connecting all vertices from $K$}

\defproblem{\textsc{Hamiltonian Cycle}}
{graph $G$}
{whether there is a cycle going through every vertex of $G$ exactly once}

\defproblem{\textsc{Feedback Vertex Set}}
{graph $G$, integer $k$}
{whether there is a set $Y \subseteq V$ of size at most $k$ such that every cycle in $G$ contains a vertex from $Y$}

In the counting variants of problems we ask for a number of structures satisfying the given conditions.
This setting is at least as hard as the decision variant.

\section{Clifford algebras}\label{sec:clif}

Some terms used in this section originate from advanced algebra.
For better understanding we suggest reading Appendix \ref{app:algebra}.

\begin{definition}
The Clifford algebra $Cl_{p,q}(R)$ is a $2^{p+q}$-dimensional associative algebra over a~ring $R$.
It is generated by $x_1, x_2 \dots, x_{p+q}$.

These are rules of multiplication of generators:
\begin{enumerate}
\item $e$ is a neutral element of multiplication,
\item $x_i^2 = e$ for $i = 1, 2, \dots, p$,
\item $x_i^2 = -e$ for $i = p+1, p+2, \dots, p+q$,
\item $x_ix_j = -x_jx_i$ if $i \ne j$.
\end{enumerate}

All $2^{p+q}$ products of ordered sets of generators form a basis of $Cl_{p,q}(R)$
($e$ is treated as a product of an empty set).
We provide a standard addition and we extend multiplication
for all elements in an associative way. 
\end{definition}

We will be mainly interested only in \clifz\footnote{Clifford algebras with $q=0$ appear also in geometric literature as \textit{exterior algebras}.} and its natural embedding into $Cl_{n,0}(\mathbb{R})$.
As $q=0$, we can neglect condition 3 when analyzing these algebras.

For $A = \{a_1, a_2, \dots, a_k\} \subseteq [1\dots n]$ where $a_1 < a_2 < \dots < a_k$
let $x_A = x_{a_1}x_{a_2}\cdots x_{a_k}$.
Each element of \clifr can be represented as $\sum_{A \subseteq [1\dots n]}a_Ax_A$, where
$a_A$ are real coefficients.
Using condition 4 we can deduce a general formula for multiplication in \clifr:

\begin{equation}\label{eq:clif-mul}
\left(\sum_{A \subseteq [1\dots n]}a_Ax_A\right)\left(\sum_{B \subseteq [1\dots n]}b_Bx_B\right)=
\sum_{C \subseteq [1\dots n]}\left(\sum_{A\triangle B = C}a_Ab_BI_{A,B}\right)x_C
\end{equation}
where the meaning of $I_{A,B}$ is explained in (\ref{eq:i}).

As a Clifford algebra over $\re$ is semisimple, it is isomorphic to a product of matrix algebras by the Artin-Wedderburn theorem
(see Theorem~\ref{thm:artin}).
However, it is more convenient to first embed \clifr in a different Clifford algebra
that is isomorphic to a single matrix algebra.
As a result, we obtain a monomorphism $\phi: Cl_{n,0}(\mathbb{R}) \longrightarrow \mathbb{M}_{2^m}(\mathbb{R})$ (see Definition~\ref{def:iso})
where $m = \frac{n}{2} + O(1)$
and the following diagram commutes ($*$ stands for multiplication).

\begin{equation}\label{eq:iso1}
\begin{array}{ccc}
Cl_{n,0}(\mathbb{R}) &\overset{\phi}{\llarrow} &\mathbb{M}_{2^m}(\mathbb{R})  \\
\downarrow * &&\downarrow *  \\
Cl_{n,0}(\mathbb{R}) &\overset{\phi}{\llarrow} &\mathbb{M}_{2^m}(\mathbb{R}) 
\end{array}
\end{equation}

Thus, we can perform multiplication in the structure that is more convenient for us.
For $a,b \in Cl_{n,0}(\mathbb{Z})$ we can treat them as elements of $Cl_{n,0}(\mathbb{R})$, find matrices $\phi(a)$ and $\phi(b)$,
multiply them efficiently, and then revert the $\phi$ transform.
The result always exists and belongs to \clifz because \clifz is closed under multiplication.
The monomorphism $\phi: Cl_{n,0}(\mathbb{R}) \longrightarrow \mathbb{M}_{2^m}(\mathbb{R})$
can be performed and reverted (within the image) in $O^*(2^n)$ time~\cite{isomorphism}.
However, the construction in~\cite{isomorphism} is analyzed in the infinite precision
model.
For the sake of completeness, we revisit this construction and prove
the following theorem in Appendix~\ref{app:clif}.

\begin{theorem}\label{thm:clifford-mul}
The multiplication in $Cl_{n,0}(\mathbb{Z})$, with coefficients having $poly(n)$ number of bits,
can be performed in time $O^*(2^{\frac{\omega n}{2}})$.
\end{theorem}

In order to unify the notation we will represent each element of \clifz, that is $\sum_{A \subseteq [1\dots n]}a_Ax_A$,
as a function $f:2^{[1\dots n]} \longrightarrow \ze,\, f(A) = a_A$.
We introduce $\diamond_S$ convolution as an equivalence of multiplication in \clifz.
The equation (\ref{eq:clif-mul}) can be now rewritten in a more compact form

\begin{equation}\label{eq:clif-mul2}
(f \diamond_S g)(X) = \sum_{A \triangle B = X}f(A)g(B)I_{A,B}.
\end{equation}

\section{Non-commutative Subset Convolution}\label{sec:nsc}

We consider a linearly ordered universe $U$ of size $n$ and functions $f,g:2^U \longrightarrow \ze$.

\begin{definition}\label{def:usc}
Let $f \diamond  g$ denote
\textit{Non-commutative Subset Convolution} (NSC) of functions $f,g$ defined as

\begin{equation*}
(f \diamond g)(X) = \sum_{A \uplus B = X}f(A)g(B)I_{A,B}.
\end{equation*}
\end{definition}

\begin{theorem}\label{the:usc}
NSC on an $n$-element universe can be performed in time $O^*(2^{\frac{\omega n}{2}})$.
\end{theorem}
\begin{proof}
Observe that condition $A \uplus B = X$ is equivalent to $A \triangle B = X \land |A| + |B| = |X|$ so

\begin{equation*}
(f \diamond g)(X) = \sum_{\substack{i + j = |X| \\ i, j \ge 0}} \sum_{A \triangle B = X}f(A)\Big[|A| = i\Big]g(B)\Big[|B| = j\Big]I_{A,B}.
\end{equation*}
Alternatively, we can write

\begin{equation*}
(f \diamond g)(X) = \sum_{\substack{i + j = |X| \\ i, j \ge 0}} (f_i \diamond_S g_j)(X),
\end{equation*}
where $f_i(X) = f(X)\Big[|X| = i\Big]$ and likewise for $g$.
The $\diamond_S$ convolution, introduced in~(\ref{eq:clif-mul2}),
is equivalent to multiplication in \clifr.
This means we reduced NSC to $O(n^2)$ multiplications
in \clifr which could be performed in time $O(2^{\frac{\omega n}{2}})$ according to Theorem~\ref{thm:clifford-mul}.
\end{proof}

\begin{observation}\label{obs:size-group}
The technique of paying polynomial factor for grouping the sizes of sets will turn useful in further proofs.
We will call it size-grouping.
\end{observation}

In our applications we will need to compute a slightly more complex convolution.

\begin{definition}\label{def:usc2}
When $f,g$ are of type $2^U \times 2^U \longrightarrow \ze$ we can define $f \diamond_2 g$
(NSC2) as follows
\begin{equation*}
(f \diamond_2 g)(X, Y) = \sum_{\substack{X_1 \uplus X_2 = X \\ Y_1 \uplus Y_2 = Y}} f(X_1,Y_1)g(X_2,Y_2)I_{X_1,X_2}I_{Y_1,Y_2}.
\end{equation*}
\end{definition}

\begin{theorem}\label{the:usc2}
NSC2 on an $n$-element universe can be performed in time $O^*(2^{\omega n})$.
\end{theorem}
\begin{proof}
Let us introduce a new universe $U' = U_X \cup U_Y$ of size $2n$ consisting of two copies
of $U$ with an order so each element of $U_Y$ is greater than any element of $U_X$.
To underline that $X \subseteq U_X, Y \subseteq U_X$ we will use $\uplus$ notation when summing subsets of $U_X$ and $U_Y$.
In order to reduce NSC2 to NSC on the universe $U'$ we need to replace factor $I_{X_1,X_2}I_{Y_1,Y_2}$
with $I_{X_1 \uplus Y_1,X_2 \uplus Y_2}$.
The latter term can be expressed as $I_{X_1,X_2}I_{Y_1,Y_2}I_{X_1,Y_2}I_{Y_1,X_2}$ due to Claim \ref{lem:i2}.
As all elements from $X_i \subseteq U_X$ compare less to elements from $Y_i \subseteq U_Y$
then $I_{X_1,Y_2} = 1$ and $I_{Y_1,X_2}$ depends only on the sizes of $Y_1$ and $X_2$.
To summarize,

\begin{equation*}
I_{X_1,X_2}I_{Y_1,Y_2} = I_{X_1 \uplus Y_1,X_2 \uplus Y_2}(-1)^{|Y_1||X_2|}.
\end{equation*}
To deal with factor $(-1)^{|Y_1||X_2|}$ we have to split the convolution
into 4 parts for different parities of $|Y_1|$ and $|X_2|$.
We define functions $f', f'_0, f'_1, g', g'_0, g'_1 : 2^{U'} \longrightarrow \ze$ as

\begin{eqnarray*}
f'(X \uplus Y) &=& f(X, Y),\\
f'_0(X \uplus Y) &=& f(X, Y)\Big[\,|Y| \equiv 0 \bmod 2\Big],\\
f'_1(X \uplus Y) &=& f(X, Y)\Big[\,|Y| \equiv 1 \bmod 2\Big],\\
g'(X \uplus Y) &=& g(X, Y),\\
g'_0(X \uplus Y) &=& g(X, Y)\Big[\,|X| \equiv 0 \bmod 2\Big],\\
g'_1(X \uplus Y) &=& g(X, Y)\Big[\,|X| \equiv 1 \bmod 2\Big].\\
\end{eqnarray*}
Now we can reduce NSC2 to 4 simpler convolutions.

\begin{eqnarray*}
(f \diamond_2 g)(X, Y) = \sum_{\substack{X_1 \uplus X_2 = X \\ Y_1 \uplus Y_2 = Y}} f'(X_1 \uplus Y_1)g'(X_2 \uplus Y_2)I_{X_1 \uplus Y_1, Y_2 \uplus X_2}(-1)^{|Y_1||X_2|} = \\
= (f'_0 \diamond g'_0)(X \uplus Y) + (f'_0 \diamond g'_1)(X \uplus Y) + (f'_1 \diamond g'_0)(X \uplus Y) - (f'_1 \diamond g'_1)(X \uplus Y)
\end{eqnarray*}
We have shown that computing NSC2 is as easy as NSC on a universe two times larger.
Using Theorem \ref{the:usc} directly gives us the desired complexity.
\end{proof}

\section{Counting Steiner trees}\label{sec:count-st}

We will revisit the theorem stated in the aforementioned work.

\begin{theorem}[Bodlaender et al. \cite{det}]\label{thm:steiner-orig}
There exist algorithms that given a graph $G$ count the number of Steiner trees
of size $i$ for each $1 \le i \le n - 1$ in $O^*(5^{pw})$ time if a path decomposition of width $pw$ is
given, and in $O^*(10^{tw})$ time if a tree decomposition of width $tw$ is given.
\end{theorem}

Both algorithms use dynamic programming over tree or path decompositions.
We introduce some decomposition-based order on $V$ and fix vertex $v_1$.
Let $A = (a_{v,e})_{v \in V, e \in E}$ be an incidence matrix, i.e.
for $e = uv,\, u < v$ we have $a_{u,e}=1,\, a_{v,e}=-1$ and $a_{w,e} = 0$ for any other vertex $w$.
For each node $x$ of the decomposition we define a function $A_x$
with arguments $0 \le i \le n - 1,\, s_Y, s_1, s_2 \in \{0, 1\}^{B_x}$.
The idea is to express the number of Steiner trees with exactly $i$ edges as $A_r(i+1,\emptyset,\emptyset,\emptyset)$.

\begin{align}\label{eq:steiner-big}
&A_x(i, s_Y, s_1, s_2) = \nonumber \\
&= \sum_{\substack{Y \subseteq V_x \\ |Y|=i \\ (K \cap V_x)\subseteq Y \\ Y \cap B_x = s_Y^{-1}(1)}}
\sum_{X\subseteq E(Y,Y)\cap E_x}
\sum_{\substack{f_1:X \overset{1-1}{\rightarrow}Y\ba\{v_1\}\ba s_1^{-1}(0) \\ f_2:X \overset{1-1}{\rightarrow}Y\ba\{v_1\}\ba s_2^{-1}(0)}}
\text{sgn}(f_1)\text{sgn}(f_2)\prod_{e\in X}a_{f_1(e),e}a_{f_2(e),e}
\end{align}

As observed in \cite{det} condition $s_Y(v)=0$ implies that either $s_1(v) = s_2(v) = 0$ or $A_x(i, s_Y, s_1, s_2) = 0$.
This means there are at most $n5^{tw}$ triples for which $A_x$ returns a nonzero value.

If a node $x$ has a child $y$ and is of type \textit{introduce vertex}, \textit{introduce edge}, or \textit{forget vertex}, then
the function $A_x$ can be computed from $A_y$ in linear time with respect to the number of non-trivial states.
Saying this is just a reformulation of Theorem \ref{thm:steiner-orig} for path decompositions.
The only thing that is more difficult for tree decompositions is that they include also \textit{join} nodes
having two children each.
Here is the recursive formula\footnote{
As confirmed by the authors~\cite{private-communication}, the formula in \cite{det} for the \text{join} node is missing the first argument to the $A_x$ function tracking the number of vertices of a Steiner tree, hence we present a corrected version of this formula.
} for $A_x$ for a \textit{join} node $x$ having children $y,z$.

\begin{equation}\label{eq:steiner-join}
A_x(i, s_Y, s_1, s_2) = \sum_{\substack{i_y + i_z = i + |s_Y^{-1}(1)| \\ s_{1,y} + s_{1,z} = s_1 \\ s_{2,y} + s_{2,z} = s_2}}
\myatop{A_y(i_y, s_Y, s_{1,y}, s_{2,y})A_z(i_z, s_Y, s_{1,z}, s_{2,z})}
{I_{s_{1,y}^{-1}(1), s_{1,z}^{-1}(1)}I_{s_{2,y}^{-1}(1), s_{2,z}^{-1}(1)}}
\end{equation}

The next lemma, however not stated explicitly in the discussed work, follows from the proof of Theorem \ref{thm:steiner-orig}
(Theorem 4.4 in \cite{det}).

\begin{lemma}\label{lem:steiner-smart}
Assume there is an algorithm computing all nonzero values of $A_x$ given by (\ref{eq:steiner-join})
with running time $f(tw)$.
Then the number of Steiner trees of size $i$ in a graph $G$ can be counted
in $O^*(\max(f(tw), 5^{tw}))$ time if a tree decomposition of width $tw$ is given.
\end{lemma}

We will change notation for our convenience.
Each function $s_i$ will be matched with a set $s_i^{-1}(1)$.
Let us replace functions $A_x, A_y, A_z$ with $h_i,f_i,g_i$ having first argument fixed and operating on triples of sets.
In this setting, the convolution can we written as

\begin{equation}\label{eq:steiner-conv1}
h_i(A,B,C) = \sum_{\substack{i_y + i_z = i + |A| \\ B_y \uplus B_z = B \\ C_y \uplus C_z = C}}
f_{i_y}(A,B_y,C_y)g_{i_z}(A,B_z,C_z)I_{B_y,B_z}I_{C_y,C_z}.
\end{equation}
Observe that size-grouping allows us to sacrifice a polynomial factor and neglect the restrictions for $i,i_y,i_z$.
Hence, we can work with a simpler formula

\begin{equation}\label{eq:steiner-conv2}
h(A,B,C) = \sum_{\substack{B_y \uplus B_z = B \\ C_y \uplus C_z = C}}
f(A,B_y,C_y)g(A,B_z,C_z)I_{B_y,B_z}I_{C_y,C_z}.
\end{equation}

The only triples $\left(s_Y(v), s_1(v), s_2(v)\right)$ allowed for each vertex $v$ are
$(0,0,0)$, $(1,0,0)$, $(1,0,1)$, $(1,1,0)$, $(1,1,1)$.
In terms of set notation we can say that if $f(A,B,C) \ne 0$ then $B \cup C \subseteq A$.
Let $f_A : 2^A \times 2^A \longrightarrow \ze$ be $f$ with the first set fixed, i.e. $f_A(B,C) = f(A,B,C)$.

\begin{lemma}\label{lem:steiner-fixed}
For fixed $A$ all values $h(A,B,C)$ can be computed in time $O^*(2^{\omega |A|})$.
\end{lemma}
\begin{proof}
We want to compute

\begin{equation*}
h_A(B,C) = \sum_{\substack{B_y \uplus B_z = B \\ C_y \uplus C_z = C}}
f_A(B_y,C_y)g_A(B_z,C_z)I_{B_y,B_z}I_{C_y,C_z} = (f_A \diamond_2 g_A)(B, C),
\end{equation*}
what can be done in time $O^*(2^{\omega |A|})$ according to Theorem \ref{the:usc2}.
\end{proof}

\begin{lemma}\label{lem:steiner-complexity}
The convolution (\ref{eq:steiner-conv1}) can be performed in time $O^*((2^{\omega} + 1)^{tw})$.
\end{lemma}
\begin{proof}
We use size-grouping to reduce the problem to computing (\ref{eq:steiner-conv2}).
Then we iterate through all possible sets $A$ and take advantage of Lemma \ref{lem:steiner-fixed}.
The total number of operations (modulo polynomial factor) is bounded by

\begin{equation*}
\sum_{A \subseteq U} 2^{\omega |A|} = \sum_{k=0}^{tw} \binom{tw}{k} 2^{\omega k} = (2^{\omega} + 1)^{tw}.
\end{equation*}
\end{proof}
Keeping in mind that  (\ref{eq:steiner-join}) and (\ref{eq:steiner-conv1}) are equivalent and
combining Lemmas \ref{lem:steiner-smart}, \ref{lem:steiner-complexity}, we obtain the following result.

\begin{theorem}\label{thm:steiner-final}
The number of Steiner trees of size $i$ in a graph $G$ can be computed
in $O^*((2^{\omega} + 1)^{tw})$ time if a tree decomposition of width $tw$ is given.
\end{theorem}

\begin{remark}
The space complexity of the algorithm is $O^*(5^{tw})$.
\end{remark}

Solving the decision version of \textsc{Feedback Vertex Set}
can be reduced to the \textsc{Maximum Induced Forest} problem \cite{det}.
As observed in \cite{det} the \textit{join} operation for \textsc{Maximum Induced Forest}
is analogous to (\ref{eq:steiner-join}).

\begin{corollary}
The existence of a feedback vertex set of size at most $i$ in a graph $G$ can be determined
in $O^*((2^{\omega} + 1)^{tw})$ time if a tree decomposition of width $tw$ is given.
\end{corollary}

\section{Counting Hamiltonian cycles}\label{sec:count-hc}

Likewise in the previous section, we will start with a previously known theorem.

\begin{theorem}[Bodlaender et al. \cite{det}]\label{thm:ham-orig}
There exist algorithms that given a graph $G$ count the number of Hamiltonian cycles
in $O^*(6^{pw})$ time if a path decomposition of width $pw$ is
given, and in $O^*(15^{tw})$ time if a tree decomposition of width $tw$ is given.
\end{theorem}

For each node $x$ of the decomposition a function $A_x$ is defined
with arguments \mbox{$s_1, s_2 \in \{0,1\}^{B_x}$}
and $s_{deg} \in \{0,1,2\}^{B_x}$.
The idea and notation is analogous to (\ref{eq:steiner-big}).
The number of Hamiltonian cycles can be expressed as $A_r(\emptyset,\emptyset,\emptyset) / n$.

\begin{align}\label{eq:ham-orig-sum}
&A_x(s_{deg},s_1,s_2)= \nonumber \\
&=\sum_{\substack{X \subseteq E_x \\ \forall_{v \in (V_x\ba B_x)} deg_X(v)=2 \\ \forall_{v\in B_x} deg_X(v)=s_{deg}(v)}}
\sum_{S\subseteq X}
\sum_{\substack{f_1:S \overset{1-1}{\rightarrow}V_x\ba\{v_1\}\ba s_1^{-1}(0) \\ f_2:S \overset{1-1}{\rightarrow}V_x\ba\{v_1\}\ba s_2^{-1}(0)}}
\text{sgn}(f_1)\text{sgn}(f_2)\prod_{e\in S}a_{f_1(e),e}a_{f_2(e),e} 
\end{align}

As observed in \cite{det} we can restrict ourselves only to some subspace of states.
When $s_Y(v)=0$ then all non-zero summands in the (\ref{eq:ham-orig-sum}) satisfy $s_1(v) = s_2(v) = 0$.
When $s_Y(v)=2$ then we can neglect all summands except for those satisfying $s_1(v) = s_2(v) = 1$.

This time there are at most $6^{tw}$ triples for which $A_x$ returns a nonzero value.
We again argue that \textit{introduce vertex}, \textit{introduce edge}, and \textit{forget vertex} nodes
can be handled the same way as for the path decomposition and the only bottleneck is formed by \textit{join} nodes.
We present a formula for $A_x$ if $x$ is a \textit{join} node with children $y, z$.

\begin{equation}\label{eq:ham-join}
A_x(s_{deg}, s_1, s_2) = \sum_{\substack{s_{deg,y} + s_{deg,z} = s_{deg} \\ s_{1,y} + s_{1,z} = s_1 \\ s_{2,y} + s_{2,z} = s_2}}
\myatop{ A_y(s_{deg,y}, s_{1,y}, s_{2,y})A_z(s_{deg,z}, s_{1,z}, s_{2,z})}
{I_{s_{1,y}^{-1}(1), s_{1,z}^{-1}(1)}I_{s_{2,y}^{-1}(1), s_{2,z}^{-1}(1)}}
\end{equation}

Analogously to the algorithm for \textsc{Steiner Tree}, we formulate our claim as a lemma following
from the proof of Theorem \ref{thm:ham-orig}
(Theorem 4.3 in \cite{det}).

\begin{lemma}\label{lem:ham-smart}
Assume there is an algorithm computing all nonzero values of $A_x$ given by (\ref{eq:ham-join})
with running time $f(tw)$.
Then the number of Hamiltonian cycles in a graph $G$ can be counted
in $O^*(\max(f(tw), 6^{tw}))$ time if a tree decomposition of width $tw$ is given.
\end{lemma}

The only allowed triples of $\left(s_{deg}(v), s_1(v), s_2(v)\right)$
for each vertex $v$ are $(0, 0, 0), $ $(1, 0, 0), $ $(1, 0, 1), $ $(1, 1, 0), $ $(1, 1, 1), $ $(2, 1, 1)$.

\begin{lemma}\label{lem:triples}
Assume the equation \ref{eq:ham-join} holds.
Then it remains true after the following translation of the set of allowed triples $\left(s_{deg}(v), s_1(v), s_2(v)\right)$.
\begin{eqnarray*}
0,0,0 \longrightarrow 0,0,0 \\
1,0,0 \longrightarrow 1,0,0 \\
1,0,1 \longrightarrow 1,0,1 \\
1,1,0 \longrightarrow 0,1,0 \\
1,1,1 \longrightarrow 0,1,1 \\
2,1,1 \longrightarrow 1,1,1 \\
\end{eqnarray*}
\end{lemma}

\begin{proof}
The $I_{.,.}$ factors do not change as we do not modify the coordinates given by functions $s_1, s_2$.
Triples that match in (\ref{eq:ham-join}) translate into matching triples as the transformation keeps their additive structure.
This fact can be seen on the tables below.

\vspace{0.5cm}
\begin{tabular}{c||c|c|c|c|c|c|}
&000 & 100 & 101 & 110 & 111 & 211 \\
\hline\hline
000 & 000 & 100 & 101 & 110 & 111 & 211 \\
\hline
100 & 100 & X & X & X & 211 & X \\
\hline
101 & 101 & X & X & 211 & X & X \\
\hline
110 & 110 & X & 211 & X & X & X \\
\hline
111 & 111 & 211 & X & X & X & X \\
\hline
211 & 211 & X & X & X & X & X \\
\hline
\end{tabular}
\,
\begin{tabular}{c||c|c|c|c|c|c|}
&000 & 100 & 101 & 010 & 011 & 111 \\
\hline\hline
000 & 000 & 100 & 101 & 010 & 011 & 111 \\
\hline
100 & 100 & X & X & X & 111 & X \\
\hline
101 & 101 & X & X & 111 & X & X \\
\hline
010 & 010 & X & 111 & X & X & X \\
\hline
011 & 011 & 111 & X & X & X & X \\
\hline
111 & 111 & X & X & X & X & X \\
\hline
\end{tabular}

\end{proof}

Therefore we can treat $s_{deg}$ functions as binary ones.
We start with unifying the notation binding functions $s_i$ with sets $s_i^{-1}(1)$.
Let us replace functions $A_x, A_y, A_z$ with their equivalences $h,f,g$ operating on triples of sets.
In this setting, the convolution looks as follows.

\begin{equation}\label{eq:hamilton1}
h(A,B,C) = \sum_{\substack{A_1 \uplus A_2 = A \\ B_1 \uplus B_2 = B \\ C_1 \uplus C_2 = C}}
f(A_1,B_1,C_1)g(A_2,B_2,C_2)I_{B_1,B_2}I_{C_1,C_2}
\end{equation}

Performing convolution (\ref{eq:hamilton1}) within the space of allowed triples is noticeably more complicated
than computations in Section \ref{sec:count-st}.
Therefore the proof of the following lemma is placed in Appendix \ref{app:hc}.

\begin{lemma}\label{lem:ham-final1}
The convolution (\ref{eq:hamilton1}) can be computed in time
$O^*((2^{\omega} + 2)^{tw})$.
\end{lemma}

This result, together with Lemmas \ref{lem:ham-smart} and \ref{lem:triples}, leads to the main theorem of this section.

\begin{theorem}\label{thm:ham-final2}
The number of Hamiltonian cycles in a graph $G$ can be computed
in $O^*((2^{\omega} + 2)^{tw})$ time if a tree decomposition of width $tw$ is given.
\end{theorem}

\begin{remark}
The space complexity of the algorithm is $O^*(6^{tw})$.
\end{remark}

\section{Conclusions}

We have presented the Non-commutative Subset Convolution, a new algebraic tool in algorithmics
based on the theory of Clifford algebras.
This allowed us to construct faster deterministic algorithms for \textsc{Steiner Tree}, \textsc{Feedback Vertex Set},
and \textsc{Hamiltonian Cycle}, parameterized by the treewidth.
As the determinant-based approach applies to all problems solvable by the Cut \& Count technique \cite{det, single},
the NSC can improve running times for a larger class of problems.

The first open question is whether the gap between time complexities for the decision and counting versions
of these problems could be closed.
Or maybe one can prove this gap inevitable under a well-established assumption, e.g. SETH?

The second question asked is if it is possible to prove a generic theorem so the lemmas
like \ref{lem:steiner-complexity} or \ref{lem:ham-final1} would follow from it easily.
It might be possible to characterize convolution algebras that are semisimple and
algorithmically construct isomorphisms with their canonical forms described by the Artin-Wedderburn theorem.

The last question is what other applications of Clifford algebras
and Artin-Wedderburn theorem can be found in algorithmics.

\textbf{Acknowledgements.} I would like to thank Marek Cygan
for pointing out the bottleneck of the previously known algorithms and
for the support during writing this paper.
I would also like to thank Paul Leopardi for helping me understand
the fast Fourier-like transform for Clifford algebras.

\bibliography{p-wlodarczyk}
\appendix

\section{Associative algebras}\label{app:algebra}

This section is not crucial to understanding the paper but it provides a bigger picture
of the applied theory.
We assume that readers are familiar with basic algebraic structures like rings or fields.
More detailed introduction can be found, e.g. in \cite{artin}.  

\begin{definition}
A linear space $A$ over a field $K$ (or, more generally, a module over a ring $K$) is called an \textit{associative algebra} if it admits a multiplication
operator $A \times A \rightarrow A$ satisfying the following conditions:
\begin{enumerate}
\item $\forall_{a,b,c \in A}\,\, a(bc)=(ab)c$,
\item $\forall_{a,b,c \in A}\,\, a(b+c)=ab+ac,\, (b+c)a = ba + ca$,
\item $\forall_{a,b \in A, k \in K}\,\, k(ab) = (ka)b = a(kb)$.
\end{enumerate}
\end{definition}

A set $W \subseteq A$ is called a \textit{generating set} if every
element of $A$ can be obtained from $W$ by addition and multiplication.
The elements of $W$ are called \textit{generators}.
It is easy to see that multiplication defined on a generating set
extends in an unambiguous way to the whole algebra.
We will often abbreviate the term \textit{associative} as we will study only such algebras.

\begin{definition}\label{def:product}
The product of algebras $A_1,A_2,\dots,A_m$ is an algebra
$A_1 \otimes A_2 \otimes \dots \otimes A_m$ with multiplication
performed independently on each coordinate.
\end{definition}

\begin{definition}\label{def:iso}
For algebras $A,B$ over a ring $K$, function $\phi:A \rightarrow B$ is called a $\textit{homomorphism of algebras}$
if it satisfy the following conditions:
\begin{enumerate}
\item $\forall_{a,b \in A}\,\, \phi(a+b) = \phi(a) + \phi(b)$,
\item $\forall_{a,b \in A}\,\, \phi(ab) = \phi(a)\phi(b)$,
\item $\forall_{a\in A, k \in K}\,\, \phi(ka) = k\phi(a)$.
\end{enumerate}
If $\phi$ is reversible within its image then we call it a $\textit{monomorphism}$
and if additionally $\phi(A) = B$ then we call $\phi$ an  $\textit{isomorphism}$
\end{definition}

Monomorphisms of algebras turn out
extremely useful when multiplication in algebra $B$ is simpler than multiplication in $A$, because
we can compute $ab$ as $\phi^{-1}\big(\phi(a)\phi(b)\big)$.
This observation is used in Theorem \ref{thm:clifford-mul}
and Lemmas \ref{lem:ham-final1}, \ref{lem:fin1}.
For a better intuition, we depict the various ways of performing multiplication
on diagrams (\ref{eq:iso1}), (\ref{eq:iso3}).

\begin{definition}
A subset $M$ of algebra $A$ is called a \textit{simple left module} if
\begin{enumerate}
\item $\forall_{a \in A, b \in M}\,\, ab \in M$,
\item $\forall_{b,c \in M}\,\, b+c \in M$,
\end{enumerate}

\noindent and the only proper subset of $M$ with these properties is $\{0\}$.
\end{definition}

The next definition is necessary to exclude some
cases of obscure algebras.

\begin{definition}
An algebra $A$ is called \textit{semisimple} if
there is no non-zero element $a$ so for every simple left module $M \subseteq A$
the set $aM = \{ab\, |\, b \in M\}$ is $\{0\}$.
\end{definition}

The theorem below was proven in full generality for algebras over arbitrary rings
but we will formulate its simpler version for fields.

\begin{theorem}[Artin-Wedderburn \cite{artin}]\label{thm:artin}
Every finite-dimensional  associative semisimple algebra $A$ over a field $K$
 is isomorphic to a product of matrix algebras

\begin{equation*}
A \cong M_{n_1}(K_1) \otimes M_{n_2}(K_2) \otimes \dots \otimes M_{n_m}(K_m),
\end{equation*}
where $K_i$ are fields containing $K$.
\end{theorem}

The related isomorphism is called a generalized Fourier transform (GFT) for $A$. 
If we are able to perform GFT efficiently then
we can reduce computations in $A$ to matrix multiplication.
For some classes of algebras, e.g. abelian group algebras~\cite{ffts},
there are known algorithms for GFT with running time $O(n\log n)$ where $n = \dim A$. 

If the field $K$ is algebraically closed (e.g. $\mathbb{C}$)
then all $K_i = K$ and $\sum_{i=1}^m n_i^2$ equals the dimension of $A$.
If the algebra $A$ is commutative then all $n_i = 1$ and $A$ is isomorphic to a product of fields.
This is actually the case in the Fast Subset Convolution \cite{fsc}
where the isomorphism is given by the M\"{o}bius transform.

\section{Proof of Theorem \ref{thm:clifford-mul}}\label{app:clif}
\begin{proof}
The transformation $\phi$ can be computed and reverted (within the image) in time $O^*(2^n)$
assuming infinite precision and $O(1)$ time for any arithmetic operation \cite{isomorphism}.
In order to compute $\phi$ accurately, we need to look inside the paper \cite{isomorphism}.

Transformation $\phi$ can be represented as $\phi = \gamma \circ \upsilon$ where
$\upsilon$ is a monomorphic embedding into another Clifford algebra and $\gamma$ is an isomorphism with the matrix algebra.
We modify isomorphism diagram (\ref{eq:iso1}) to show these mappings in more detail.

\begin{equation*}
\begin{array}{ccccccc}
Cl_{n,0}(\mathbb{Z}) &\hookrightarrow  & Cl_{n,0}(\mathbb{R}) &\overset{\upsilon}{\longrightarrow} &Cl_{m,m}(\mathbb{R}) &\overset{\gamma}{\longrightarrow} &\mathbb{M}_{2^m}(\mathbb{R})  \\
\downarrow * &&\downarrow * &&\downarrow * &&\downarrow *\\
Cl_{n,0}(\mathbb{Z}) &\hookrightarrow  & Cl_{n,0}(\mathbb{R}) &\overset{\upsilon}{\longrightarrow} &Cl_{m,m}(\mathbb{R}) &\overset{\gamma}{\longrightarrow} &\mathbb{M}_{2^m}(\mathbb{R}) 
\end{array}
\end{equation*}

We begin with embedding $\upsilon: Cl_{n,0}(\mathbb{R}) \longrightarrow Cl_{m,m}(\mathbb{R})$ where $m = \frac{n}{2} + O(1)$
(see Definition 4.4 in \cite{isomorphism}).
Transformation $\upsilon$ is just a translation of basis so no arithmetic operations are required.

For the sake of disambiguation, we indicate the domain of the function $\gamma$ with a lower index:
	$\gamma_k: Cl_{k,k}(\mathbb{R}) \longrightarrow \mathbb{M}_{2^k}(\mathbb{R})$.
In the $k$-th step, we construct a matrix representation of $y \in Cl_{k,k}(\re)$.
Let $y^+, y^-$ denote the projections of $y$ onto subspaces spanned by products of respectively even and odd number of generators.
	Of course, $y = y^+ + y^-$ and \mbox{$\gamma_k(y) = \gamma_k(y^+) + \gamma_k(y^-)$}.
	Such an element $y$ can be represented as $y = a + b\mathbf{x_-} + c\mathbf{x_+} + d\mathbf{x_-}\mathbf{x_+}$
	for $\mathbf{x_+},\mathbf{x_-}$ being the first and the last generator ($\mathbf{x}_+^2 = e,\, \mathbf{x}_-^2 = -e$) and $a,b,c,d \in Cl_{k-1,k-1}(\re)$.
	Now we can apply the recursive formula from Theorem~5.2~in~\cite{isomorphism}:

	\begin{equation*}
		\gamma_k(y^+) = \gamma_{k-1}\left(\left[ \begin{array}{cc} a^+ - d^+ & -b^- - c^- \\ -b^- - c^- & a^+ + d^+ \end{array} \right]\right),\quad \gamma_k(y^-) = \gamma_{k-1}\left(\left[ \begin{array}{cc} a^- - d^- & -b^+ + c^+ \\ b^+ + c^+ & -a^- - d^- \end{array} \right]\right),
	\end{equation*}
	where $\gamma_{k-1}(M)$ stands for a block matrix with $\gamma_{k-1}$ applied to each element of $M$.

We see that computing $\big(\gamma_k(y^+), \gamma_k(y^-)\big)$ can be reduced to computing 4 analogous pairs for $k-1$
and combining them using addition and subtraction.
Hence, the coefficients of the obtained matrix will also be integers with $poly(n)$ number of bits
	and the total number of arithmetic operations is $O(m4^m) = O(n2^n)$.

The inverse transform $\gamma^{-1}$ is also computed in $m$ steps and we continue using lower
index to indicate the domain alike for the forward transform.
	Let $Y \in \mathbb{M}_{2^k}(\ze)$ and

	\begin{equation*}
		Y = \left[ \begin{array}{cc} Y_{11} & Y_{12} \\ Y_{21} & Y_{22} \end{array} \right],\quad y_{ij} = \gamma_{k-1}^{-1}(Y_{ij}).
	\end{equation*}
Then from Theorem 7.1 in \cite{isomorphism} we know that

	\begin{equation*}
		\gamma_{k}^{-1}(Y) = \frac{1}{2}\big((\hat{y_{22}} + y_{11}) + (\hat{y_{21}} - y_{12})\mathbf{x_-} + (\hat{y_{21}} + y_{12})\mathbf{x_+} + (\hat{y_{22}} - y_{11})\mathbf{x_-}\mathbf{x_+}\big),
	\end{equation*}
	where $\hat{y} = y^+ - y^-$ and the rest of notation is as above.
	We can reduce computing $\gamma_{k}^{-1}$ to 4~queries from $(k-1)$-th step so the total number
	of arithmetic operations is $O(m4^m) = O(n2^n)$.

This time the coefficients at each step are given as sums of elements from the previous step
divided by 2.
We do not need to prove that they remain integer at all steps because we can postpone the division
until the last step.
As long as $\gamma^{-1}(Y)$ is a product of two elements from $Cl_{m,m}(\mathbb{Z})$, it is guaranteed
that the numbers in the last step would be divisible by $2^m$.
What is more, if we know that $\gamma^{-1}(Y) \in \upsilon(Cl_{n,0}(\mathbb{Z}))$ then
we can revert the $\upsilon$ transform and obtain $\phi^{-1}(Y)$.
	
	We have proven that we can switch representation between $Cl_{n,0}(\ze)$ and $\mathbb{M}_{2^m}(\ze)$
	in time $O^*(2^n)$.
	The multiplication in $\mathbb{M}_{2^m}(\mathbb{Z})$ for inputs of $poly(n)$ size can be performed in time complexity
$O^*(2^{\omega m}) = O^*(2^{\frac{\omega n}{2}})$ and the resulting matrix also contains
only $poly(n)$-bits integers.
	This proves that the multiplication in $Cl_{n,0}(\ze)$ admits an algorithm with running time $O^*(2^{\frac{\omega n}{2}})$.

\end{proof}

\section{Proof of Lemma \ref{lem:ham-final1}}\label{app:hc}

This section reduces the complicated algorithm for \textsc{Hamiltonian Cycle} to
two isomorphism theorems and we suggest reading Appendix \ref{app:algebra} first.
Our goal is to compute values of $h$ for the allowed triples assuming that non-zero values of $f,g$ also occur only for the allowed triples.

\begin{equation}\label{eq:hamilton2}
h(A,B,C) = \sum_{\substack{A_1 \uplus A_2 = A \\ B_1 \uplus B_2 = B \\ C_1 \uplus C_2 = C}}
f(A_1,B_1,C_1)g(A_2,B_2,C_2)I_{B_1,B_2}I_{C_1,C_2}
\end{equation}
Taking advantage of the size-grouping technique (see Observation \ref{obs:size-group})
we can replace condition $A_1 \uplus A_2 = A$ with $A_1 \cup A_2 = A$
and focus on the following convolution.

\begin{equation}\label{eq:hamilton3}
(f \odot g)(A,B,C) = \sum_{\substack{A_1 \cup A_2 = A \\ B_1 \uplus B_2 = B \\ C_1 \uplus C_2 = C}}
f(A_1,B_1,C_1)g(A_2,B_2,C_2)I_{B_1,B_2}I_{C_1,C_2}
\end{equation}

Let $Ham$ be a subspace of $2^U \times 2^U \times 2^U \longrightarrow \ze$
given by functions admitting only the allowed triples (see Lemma \ref{lem:triples}), i.e. $f \in Ham \land f(A,B,C)\neq 0$
implies $A \cap (B \triangle C) = C \ba B$.
Observe that $Ham$ is closed under the $\odot$ operation so it can be regarded as
a $6^{tw}$-dimensional algebra.
Let $H_D$ be an algebra over space $2^{U\ba D} \times 2^D \times 2^D \longrightarrow \ze$
with multiplication given by the $\oslash$ operator defined as

\begin{equation*}
(f \oslash g)(E,B,C) = \sum_{\substack{E_1 \uplus E_2 = E \\ B_1 \uplus B_2 = B \\ C_1 \uplus C_2 = C}}
f(E_1,B_1,C_1)g(E_2,B_2,C_2)I_{B_1,B_2}I_{C_1,C_2}(-1)^{|E_1|(|B_2|+|C_2|)}.
\end{equation*}
We want to show that $Ham$ is isomorphic (see Definition~\ref{def:iso}) with a product of all $H_D$ for $D \subseteq U$
(see Definition~\ref{def:product}).
In particular, diagram~(\ref{eq:iso3}) commutes.

\begin{align}
\begin{array}{ccc}\label{eq:iso3}
Ham &\overset{\tau}{\llarrow} &\bigotimes\limits_{D \subseteq U}  H_D\\
\downarrow \odot &&\quad\quad\quad\downarrow \oslash  \\
Ham &\overset{\tau}{\llarrow} &\bigotimes\limits_{D \subseteq U}  H_D
\end{array}
\end{align}
where $\tau_D: Ham \longrightarrow H_D$ is given as
\begin{equation*}
(\tau_D f)(E,B,C) = I_{B,E}I_{C,E}\sum_{A \subseteq D} f(A,B\cup E,C\cup E).
\end{equation*}

\begin{lemma}\label{lem:fin1}
Transform $\tau$ and its inverse can be performed in time $O^*(6^{tw})$.
\end{lemma}

\begin{corollary}
Transformation $\tau$ is reversible.
\end{corollary}

\begin{lemma}\label{lem:fin2}
Given $f, g \in H_D$ we can compute $ f\oslash g$
in time $O^*(2^{\omega|D|}2^{|U\ba D|})$.
\end{lemma}

\begin{lemma}\label{lem:diagram}
Diagram (\ref{eq:iso3}) commutes, i.e. $\tau$ is a homomorphism of algebras.
\end{lemma}

\begin{corollary}
Transformation $\tau$ is an isomorphism of algebras.
\end{corollary}

As for the Clifford algebras, we can switch the representation of the algebra to
perform multiplication in the simpler one, and then revert the isomorphism to get the result.
The most time consuming part of the algorithm is performing the $\oslash$ convolutions.
Total number of operations modulo polynomial factor can be bounded with Lemma \ref{lem:fin2} by

\begin{equation}
\sum_{D \subseteq U} 2^{\omega|D|}2^{|U\ba D|} = \sum_{k=0}^{tw} \binom{tw}{k} 2^{\omega k}2^{tw-k} = (2^{\omega} + 2)^{tw}.
\end{equation}
The rest of the appendix is devoted to proving Lemmas \ref{lem:fin1}, \ref{lem:fin2}, \ref{lem:diagram}.

\begin{proof}[Proof of Lemma \ref{lem:fin1}]
For fixed sets $B,C$ let $H = B \cap C,\, F = B \triangle C,\, B_1 = B \ba C,\, C_1 = C \ba B$.
Observe that every allowed triple $(A, B, C)$ must satisfy $A \cap F = C_1$.
Therefore we can represent $Ham$ as a union of sets

\begin{align*}
T_{B_1,C_1,H} = \Big\{(A_1 \cup C_1, B_1 \cup H, C_1 \cup H)\,\Big|\,A_1 \subseteq U\ba (B_1 \cup C_1)\Big\}.
\end{align*}
for all pairwise disjoint triples $B_1, C_1, H \subseteq U$.
Functions over $T_{B_1,C_1,H}$ can be parameterized with only the $A_1$ argument.
Consider following transformation over function space on $T_{B_1,C_1,H}$.

\begin{align*}
(\gamma_{B_1,C_1,H} f)(A_1) = \sum_{A_0 \subseteq A_1} f(A_0 \cup C_1, B_1 \cup H, C_1 \cup H)
\end{align*}

Transform $\gamma_{B_1,C_1,H}$ is just the M\"{o}bius transform,
therefore it can be performed and reverted in time $O^*(2^{|U\ba (B_1 \cup C_1)|})$ (see Theorem \ref{thm:fsc}).
Values of $\gamma f$ correspond directly to values of $\tau f$.

\begin{align*}
(\tau_D f)(E,B,C) &= I_{B,E}I_{C,E}\sum_{A \subseteq D} f(A,B\cup E,C\cup E) = \\
&= I_{B,E}I_{C,E}\sum_{A \subseteq D} f(A,B_1\cup H \cup E, C_1\cup H\cup E) = \\
&= I_{B,E}I_{C,E}\sum_{A_0 \subseteq D \ba F} f(A_0 \cup C_1 ,B_1\cup H \cup E, C_1\cup H\cup E) = \\
&= I_{B,E}I_{C,E}(\gamma_{B_1,C_1,H\cup E} f)(D \ba F) \\ \\
(\gamma_{B_1,C_1,H} f)(A_1) &= \sum_{A_0 \subseteq A_1} f(A_0 \cup C_1, B_1 \cup H, C_1 \cup H) =\\
&= \sum_{A_0 \subseteq A_1 \cup C_1} f(A_0, B_1 \cup H, C_1 \cup H) =\\
&= \sum_{A_0 \subseteq A_1 \cup C_1} f\big(A_0, B_2 \cup (H \ba A_1), C_2 \cup (H \ba A_1)\big) =\\
&=(\tau_{A_1 \cup C_1} f)(E,B_2,C_2)I_{B_2,E}I_{C_2,E}
\end{align*}
where $E = H\ba A_1,\, B_2 = B_1 \cup (H \cap A_1),\, C_2 = C_1 \cup (H \cap A_1)$ are valid arguments of $\tau_{A_1 \cup C_1}$.

To estimate the total number of operations consider all choices of $F$.
The partition into $F = B_1 \uplus C_1$ can be done in $2^{|F|}$ ways,
the set $H$ can be chosen in $2^{|U\ba F|}$ ways,
and for such triple we have to perform the $\gamma_{B_1,C_1,H}$ transform
(or its inverse) what involves $O^*(2^{|U\ba F|})$ operations.
Hence, the total running time (modulo polynomial factors) is
\begin{equation*}
\sum_{F \subseteq U} 2^{|F|}4^{|U\ba F|} = \sum_{k=0}^{tw} \binom{tw}{k} 2^{k}4^{tw-k} = 6^{tw}.
\end{equation*}
\end{proof}

\begin{proof}[Proof of Lemma \ref{lem:fin2}]

Applying the size-grouping (see Observation \ref{obs:size-group}) allows us to neglect the $(-1)^{|E_1|(|B_2|+|C_2|)}$ factor
and replace condition $E_1 \uplus E_2 = E$ with $E_1 \cup E_2 = E$.
Therefore it suffices to perform the $\odot$ convolution on $H_D$ (the same as in (\ref{eq:hamilton3})).

\begin{equation*}
(f \odot g)(E,B,C) = \sum_{\substack{E_1 \cup E_2 = E \\ B_1 \uplus B_2 = B \\ C_1 \uplus C_2 = C}}
f(E_1,B_1,C_1)g(E_2,B_2,C_2)I_{B_1,B_2}I_{C_1,C_2}.
\end{equation*}
Let us denote

\begin{equation*}
(\mu_E f)(B,C) = \sum_{F \subseteq  E} f(F, B, C).
\end{equation*}

Transform $\mu$ and its inverse can be computed using M\"{o}bius transform (see Theorem~\ref{thm:fsc})
in time $O^*(2^{|U\ba D|})$ for all $E$ and a fixed pair of sets $B,C$.
We perform it for all $4^{|D|}$ such pairs.

It turns out that $\mu$ is an isomorphism between $(H_D, \odot)$ and a product of
all algebras given by images of $\mu_E$ for $E \subseteq U\ba D$ (see Definitions~\ref{def:product},~\ref{def:iso})
with multiplication given by NSC2, i.e. $(\mu_E f) \diamond_2 (\mu_E g) = \mu_E (f \odot g)$.
We can again switch the representation of the algebra, multiply the elements, and then revert the isomorphism.
The computations below show that $\mu$ is a homomorphism of algebras and we know already that $\mu$ is reversible.

\begin{align*}
\big((\mu_E f) \diamond_2 (\mu_E g)\big)(B, C) &= \\
&= \sum_{\substack{B_1 \uplus B_2 = B \\ C_1 \uplus C_2 = C}}
(\mu_E f)(B_1,C_1)(\mu_E g)(B_2,C_2)I_{B_1,B_2}I_{C_1,C_2} = \\
&= \sum_{\substack{E_1,E_2 \subseteq E \\ B_1 \uplus B_2 = B \\ C_1 \uplus C_2 = C}}
f(E_1,B_1,C_1)g(E_2,B_2,C_2)I_{B_1,B_2}I_{C_1,C_2} = \\
&= \sum_{F \subseteq E} \sum_{\substack{E_1 \cup E_2 = F \\ B_1 \uplus B_2 = B \\ C_1 \uplus C_2 = C}}
f(E_1,B_1,C_1)g(E_2,B_2,C_2)I_{B_1,B_2}I_{C_1,C_2} = \\
&=\big(\mu_E (f \odot g)\big)(B, C) &
\end{align*}

To perform multiplication of $\mu(a)$ and $\mu(b)$, where $a,b \in H_D$,
we have to perform NSC2 \big($O^*(2^{\omega |D|})$ time complexity, see Theorem \ref{the:usc2}\big)
for each $E \subseteq U\ba D$, what results in desired running time.
\end{proof}

\begin{proof}[Proof of Lemma \ref{lem:diagram}]
We need to show that for each $B, C\subseteq D, D \cap E = \emptyset$
it is $(\tau_D (f \odot g))(E,B,C) = ((\tau_D f) \oslash (\tau_D g))(E,B,C)$.
Let us start with unrolling the formula for $\tau_D (f \odot g)$.
Keeping in mind that $B \cap E = C \cap E = \emptyset$ we can see that
\begin{align}
&(\tau_D (f \odot g))(E,B,C) = \nonumber \\
&=\sum_{\substack{\quad A \subseteq D\quad}}(f \odot g)(A,B\cup E, C \cup E)I_{B,E}I_{C,E}=\nonumber \\
&= \sum_{\substack{A_1, A_2 \subseteq D \\ B_1 \uplus B_2 = B \\ E_1 \uplus E_2 = E \\ C_1 \uplus C_2 = C \\ F_1 \uplus F_2 = E}}
\myatop{f(A_1,B_1 \cup E_1,C_1\cup F_1)g(A_2,B_2\cup E_2,C_2\cup F_1)}
{I_{B_1 \cup E_1, B_2\cup E_2}I_{C_1\cup F_1, C_2\cup F_2}I_{B,E}I_{C,E}.} \label{eq:diag-i1}
\end{align}
On the other hand, we have

\begin{align}
&((\tau_D f) \oslash (\tau_D g))(E,B,C) = \nonumber \\
&=\sum_{\substack{E_1 \uplus E_2 = E \\ B_1 \uplus B_2 = B \\ C_1 \uplus C_2 = C}}
(\tau_D f)(E_1,B_1,C_1)(\tau_D g)(E_2,B_2,C_2)I_{B_1,B_2}I_{C_1,C_2}(-1)^{|E_1|(|B_2|+|C_2|)}=\nonumber \\
&= \sum_{\substack{A_1, A_2 \subseteq  D \\ E_1 \uplus E_2 = E \\ B_1 \uplus B_2 = B \\ C_1 \uplus C_2 = C}}
\myatop{f(A_1,B_1\cup E_1,C_1\cup E_1)g(A_2,B_2\cup E_2,C_2\cup E_2)}
{I_{B_1,B_2}I_{C_1,C_2}I_{B_1,E_1}I_{C_1,E_1}I_{B_2,E_2}I_{C_2,E_2}(-1)^{|E_1|(|B_2|+|C_2|)}.} \label{eq:diag-i2}
\end{align}

We want to argue that all non-zero summands of (\ref{eq:diag-i1}) satisfy $E_1 = F_1, E_2 = F_2$.
Indeed, let us assume $v \in F_1 \ba E_1$.
As $v \in E$ so $v \not\in D \supseteq A, B, C$ and $\big([v \in A_1], [v \in B_1 \cup E_1],$ $[v \in C_1 \cup F_1]\big) = (0, 0, 1)$
which is not a valid triple what implies $f(A_1,B_1 \cup E_1,C_1\cup F_1) = 0$.

Assumption $v \in E_1 \ba F_1$ leads to $\big([v \in A_1], [v \in B_1 \cup E_1], [v \in C_1 \cup F_1]\big) = (0, 1, 0)$ but
$v \in E = E_1 \uplus E_2 = F_1 \uplus F_2$
so $\big([v \in A_2], [v \in B_2 \cup E_2], [v \in C_2 \cup F_2]\big) = (0, 0, 1)$
and $g(A_2,B_2\cup E_2,C_2\cup F_1) = 0$.
The same arguments can be used if $v \in E_2 \triangle F_2$.

Now we just need to prove that for $E_1 = F_1, E_2 = F_2$ the $I$ factors in (\ref{eq:diag-i1}) and (\ref{eq:diag-i2}) are equivalent.
We apply Claim \ref{lem:i2} to $I_{B_1 \cup E_1, B_2\cup E_2}I_{C_1\cup E_1, C_2\cup E_2}$.
We can omit factor $I^2_{E_1,E_2} = 1$ as well as $I_{B_1,B_2}I_{C_1,C_2}$ appearing also in (\ref{eq:diag-i2}).
What is left to prove is that

\begin{eqnarray*}
&I_{B_1,E_2}I_{E_1,B_2}I_{B,E} = I_{B_1,E_1}I_{B_2,E_2}(-1)^{|E_1||B_2|}, \\
&I_{C_1,E_2}I_{E_1,C_2}I_{C,E} = I_{C_1,E_1}I_{C_2,E_2}(-1)^{|E_1||C_2|}.
\end{eqnarray*}

According to Claim \ref{lem:i1} we can replace $I_{E_1,B_2}(-1)^{|E_1||B_2|}$ with $I_{B_2,E_1}$
what reduces the formula in the first row to Claim \ref{lem:i2} for $B = B_1 \uplus B_2,\, E = E_1 \uplus E_2$.
Applying analogous observation to the second row finishes the proof.
\end{proof}

\end{document}